\documentclass[11pt]{article}
\usepackage{}

\usepackage{pgf}
\usepackage[english]{babel}
\usepackage[utf8]{inputenc}
\usepackage{booktabs}
\usepackage{caption}
\usepackage[T1]{fontenc}
\usepackage[font=small,labelfont=bf,tableposition=top]{caption}
\usepackage{threeparttable}

\usepackage{datetime}
\usepackage{geometry}
\usepackage{array}
\usepackage{xspace}

\usepackage{amsmath,amsfonts,amssymb,amsopn,amsthm,amscd}
\theoremstyle{plain}
\usepackage{graphicx} 
\usepackage{epstopdf}
\usepackage{enumerate}

\def\Box{\vcenter{\vbox{\hrule\hbox{\vrule
     \vbox to 8.8pt{\hbox to 10pt{}\vfill}\vrule}\hrule}}}

\usepackage{indentfirst}
\newcommand{\Ff}{{\mathbb F}}
\newcommand{\F}{{\mathbb F}}

\newcommand\cc{{\mathcal C}}        %

\def\mC{{\mathcal C}}

\def\mC{{\mathcal C}}

\def\Ga{{\alpha}}

\def\bc{{\bf c}}

\def\bc{{\bf c}}

\def\GRS{{\rm GRS}}
\def\Tr{\operatorname{Tr}}

\newtheorem{thm}{Theorem}[section]
\newtheorem{lem}[thm]{Lemma}

\newtheorem{cor}[thm]{Corollary}

\newtheorem{prop}[thm]{Proposition}

\usepackage{todonotes}

\begin{document}

%

\title{Optimal repairing schemes for Reed-Solomon codes with alphabet sizes linear in lengths under the rack-aware model}

\author{ Lingfei Jin, Gaojun Luo and Chaoping Xing}

\maketitle

\let\thefootnote\relax\footnotetext{L. Jin is with Shanghai Key Laboratory of Intelligent Information Processing, School of Computer Science, Fudan  University, Shanghai 200433, P. R. China. Fudan-Zhongan Joint Laboratory of Blockchain and Information Security, Shanghai Engineering Research Center of Blockchain, and Shanghai Institute of Intelligent Electronics $\&$ Systems, Shanghai, P.R. China, Email: {lfjin@fudan.edu.cn }}

\let\thefootnote\relax\footnotetext{G. Luo is with Department of Mathematics, Nanjing University of Aeronautics and Astronautics, Nanjing 211100, China. Email: {gjluo1990@163.com}}

\let\thefootnote\relax\footnotetext{C. Xing is with School of Electronic Information and Electrical Engineering,
Shanghai Jiao Tong University, Shanghai 200240, China. He is also with
School of Physical and Mathematical Sciences, Nanyang Technological University, Singapore. Email:
{xingcp@ntu.edu.sg}}


\begin{abstract}
In modern practical data centers, storage nodes are usually organized into equally sized groups, which is called racks. The cost of cross-rack communication is much more expensive compared with the intra-rack communication cost. The codes for this system are called rack-aware regenerating codes. Similar to standard minimum storage regenerating  (MSR) codes, it is a challenging task to construct minimum storage rack-aware regenerating (MSRR) codes achieving the cut-set bound. The known constructions of MSRR codes achieving the cut-set bound give codes with alphabet size $q$ exponential in the code length $n$, more precisely, $q=\Omega(\exp(n^n))$.

The main contribution of this paper is to provide explicit construction of  MSRR codes achieving the cut-set bound with the alphabet size  linear in $n$. To achieve this goal, we first present a general framework to repair Reed-Solomon codes. It turns out that the known repairing schemes of Reed-Solomon codes can be realized under our general framework. Several techniques are used in this paper. In particular, we use the degree decent method to repair failure node. This technique allows us to get Reed-Solomon codes with the alphabet size  linear in $n$. The other techniques include choice of good polynomials. Note that good polynomials are used for construction of locally repairable code in literature. To the best of our knowledge, it is the first time in this paper to make use of good polynomials for constructions of regenerating codes.

{\bf Keywods}: Rack-aware model, Reed-solomon codes, minimum storage regenerating codes.

\end{abstract}

\section{Introduction}

In a distributed storage system where data is written in a large number of physical storage nodes, failure of a node or few nodes render a portion of the data inaccessible. Erasure error correcting codes are widely used as a coding technique with good reliability and low storage redundancy compared with replication.  When a node fails, one need to recover the information in the failed node by connecting few other active nodes. Therefore, the repair bandwidth is the total amount of the information that one need to download to complete the repair procedure.The repair bandwidth and storage overhead are two important metrics for distributed storage scheme.

For a data file, we usually divide the original file into $k$ blocks where each block is an element of a finite field $F$ or  a vector over $F$. Among all those erasure codes used in distributed storage, Reed-Solomon (RS for short) codes are widely deployed which are maximum distance separable (MDS) codes. We refer an $[n,k]$ RS code as a code with length $n$ and dimension $k$. An $[n,k]$ RS code encodes a data file of $k$ blocks into $n$ blocks by adding $n-k$ redundancy and then distribute $n$ blocks into $n$ physical storage nodes.

The fact that single node failure is the most common scenario makes the problem of repairing one node failure as the most interesting problem. Therefore, we will focus on the recovery of one failure node. If we encode the original file using a RS code, the conventional scheme is to  download any $k$ active nodes to repair the failure node. In fact, the MDS property guarantees that one can recover the whole data by downloading any $k$ nodes.
However,  this naive  repairing scheme has bandwidth  much larger than waht is needed for recovering one failure node and hence it  is not efficient for this scenario.

To minimize the repair bandwidth in the repairing procedure, the concept of regenerating codes (RC) was formulated  \cite{DG10}. It was shown in \cite{DG10} that there is a trade-off between storage and repair bandwidth. Codes lying over this trade-off are called {\it regenerating codes}. There are two special cases of regenerating codes that are interesting from the theoretical point of view. One is called minimum bandwidth regenerating (MBR for short) code where the minimum repair bandwidth is needed to repair the failed nodes. The other case is called minimum storage regenerating (MSR for short) code that corresponds to the minimum storage. We refer the reference \cite{BKV} for an excellent survey on regenerating codes.

There are many follow-up studies on regenerating codes. In \cite{SP14}, the repair process for Reed-Solomon codes was studied. A clever idea was introduced by Guruswami and Wootters \cite{SP14} using the trace function to repair any single erasure for Reed-Solomon codes, where one only needs to download partial information from other helper nodes instead of downloading the whole date in those nodes. Thus the repair bandwidth can be reduced. This idea was further generalized by Tamo et al. \cite{Tamo}. The problem of recovering multiple erased  nodes for RS codes was also considered in \cite{Hoang}.

However, classical regenerating codes are still limited in addressing the hierarchical nature of data centers. In modern practical data centers, storage nodes are usually organized into equally sized groups, which is called racks. Compared with the intra-rack communication cost, the cross-rack communication cost is much more expensive. Therefore, within each group nodes can communicate freely without taxing the system bandwidth and only the information transmission from other racks counts. 
 We can see that regenerating codes can be considered as a special case in this rack-based model where each rack only contains one node. If each rack contains more than one node, then the model distinguishes the communication costs between intra-rack and cross-rack. Thus, classical regenerating codes can not minimize the cross-rack repair bandwidth.
Therefore, exploring the trade-off between storage overhead and cross-rack repair bandwidth is of great interests in rack-aware model. It is easy to observe that the trade-off of rack-aware regenerating codes can be degenerated to that for regenerating codes if each rack only has one node.


\subsection{System model}
Assume a file of size $M$ is divided into $k$ blocks and then encoded  using an error correcting code $\cc$ to  $t n$ symbols in a finite filed $F$. In the rack-based model, we assume that the system center consists of $n$ nodes which are equally divided into racks of size $u$. Thus there are $r$ racks in total (here we assume $ur=n$ ). If one node fails, we can repair it by connecting at most $d$ racks ($d\le r-1$). Those racks participating in the repair procedure is called {\it helper racks} and $d$ is called the {\it repair degree}. The rack where a node fails is called the {\it host rack}.

A single node repair works as follows. Let $X_{h,i}$ be the $i$-th node in the $h$-th rack, $h=1,\cdots,r$, $i=1,\cdots,u$. Without loss of generality, assume that the data in $X_{1,1}$ is erased. We select a special node $X_{h,1}$ which is called a {\it relayer} of rack $h$. A relayer can connect all the data in all the available nodes in the same rack. Hence, if a data collector connects to a relayer, then it connects to all the other nodes in the same rack. The repair procedure is to generate a new node to store the lost data in $X_{1,1}$ in the host rack. The repair process has two steps. In the first step, we select any $d$ helper racks and collect the relayers of these racks. In the second step, we collects data from all the other surviving nodes  in rack 1, i.e., $X_{1,2},\cdots, X_{1,u}$. Then we can regenerate the lost data in $X_{1,1}$ by all the data from the $u-1$ nodes in rack 1 and from the other $d$ ralayers from $d$ helper racks.

In this rack-based model, we can ignore the intra-rack communication cost and specifically focus on the minimization of cross-rack repair bandwidth (i.e., the total amount of information downloaded from the other racks during a repair process).
Furthermore the nodes within the same  rack can collaborate locally before being downloaded.


\subsection{Known results}

The rack model was previously discussed in \cite{PY13,CB18,TC14,HL18,Hu17}. Though the communication costs between intra-rack and cross-racks are distinguished, the system model are not always the same. 
Double regenerating codes (DRC) was proposed by Hu et al. \cite{Hu16} to minimize the cross-rack repair bandwidth assuming that the minimum storage is achieved . The idea is rebuilding the failed symbol partially within each rack and then combine them to recover the symbol across racks. It was illustrated  that the cross-rack repair bandwidth of double regenerating codes can be less than that of regenerating codes for some code parameters. In fact, DRC are a special example of  MSRR codes with all the other $r-1$ racks being helper racks to repair a failure node. The minimum storage codes and minimum bandwidth code were also investigated in \cite{PAM18}. But the parameter $k$ must be a multiple of the number of nodes in each rack in their model.

In a recent paper by Hou et al. \cite{HL18}, this model was studied both for minimum-storage (MSR) and minimum bandwidth (MBR) scenarios while they have more flexility on the choice of parameters. It was also shown that there exists codes with optimal repair bandwidth for a wide-range of parameter. Similar to regenerating codes, exploring the tradeoff between storage redundancy and cross-rack repair bandwidth is an interesting problem. Until now, few explicit constructions of MSR codes are known for this model. A construction of minimum storage rack-aware regenerating (MSRR) codes was given in  \cite{Hu16}, where the underlying field is of size at most $n^2/u$ ($n$ is the block length, $u$ is the size of the rack). Very recently, Chen and Barg \cite{CB19} constructed MSRR codes for all parameters by designing suitable parity-check equations. Furthermore, they also proposed a construction of MSRR codes from RS codes. However, the alphabet size of RS codes given in \cite{CB19} is exponential in the code length $n$. Thus reducing the alphabet size for MSRR codes which can achieve the cut-set bound is a challenging task.


\subsection{Our contributions and techniques}

The main contribution of this paper is to provide explicit construction of several classes of MSRR codes achieving the cut-set bound with the alphabet size  linear in $n$.
The  result and technique of this paper can be summarized as follows.
\begin{itemize}
\item[1.] Utilizing punctured Reed-Solomon codes, we provide a new general framework on repairing Reed-Solomon codes in the rack-aware model where one can repair a failed node with the help of any other $d$ helper racks. It turns out that this repairing scheme is degenerated to a repairing scheme for regenerating codes if each rack contains only one node. We show that many of previous repairing schemes for Reed-Solomon codes are examples of our general framework. Furthermore, we also show that Chen and Barg's \cite{CB19} repairing scheme for Reed-Solomon codes in the rack-aware model can be  realized under our general framework.
\item[2.] We present three classes of Reed-Solomon codes in the rack-aware model with repairing schemes achieving the rack-aware cut-set bound and the field size is linear in the length of these codes. In order to do so, we employ good polynomials introduced in \cite{LRC1} to design the Reed-Solomon codes. Good polynomials, as a class of well known polynomials, are widely used in the constructions of optimal locally recoverable codes \cite{LRC1,LRC2,LRC3} and symmetric cryptography \cite{ME1}. In this paper, we first use good polynomials to design repairing schemes for Reed-Solomon codes in the rack-aware model which can deduce the cross-rack repair bandwidth. Due to use of good polynomials, the degree of the polynomial  that we design to repair the failed node is very high and exceed the field size of Reed-Solomon codes. To obtain a codeword in the dual code of a Reed-Solomon code, we have to reduce a higher degree polynomial to a lower degree polynomial via polynomial degree decent method. This is a fresh new idea that people have never used in this topic. This is also why we can get a code with length linear in the field size. Notably, the alphabet size of Reed-Solomon codes from Chen and Barg's construction is exponential in length $n$, i.e,  $\Omega(\exp({n^n}))$. 
\end{itemize}
In summary, in this paper, we (i) provide a framework to construct MSRR codes; (ii) explore the idea of good polynomials in our construction; (iii) make use of degree decent method.

\subsection{	Organization of this paper}
The rest of the paper is organized as follows. Section II provides some backgrounds on finite fields, Reed-Solomon codes and rack-aware regenerating codes. In section III, we provide a general framework on repairing Reed-Solomon codes in the rack-aware model. In section IV, we run several examples of RS codes under the general framework. In the last section, we employ some good polynomials satisfying the conditions required in the general framework to obtain repairing scheme for RS codes meeting the rack-aware cut-set bound.

\section{Preliminaries}

In this section, we recall some relevant definitions and notations used in the following sections.

\subsection{Background on finite fields}
Denote by $[n]$ the set $\{1,2,\dots,n\}$. Let $p$ be a prime power and $\Ff_p$ a finite filed with $p$ elements. Assume that $\F_q/\F_p$ is a field extension with $[\F_q:\F_p]=t$.
The trace function is a  map from $\F_q$ to $\F_p$ defined by
\[\Tr_{\F_q/\F_p}(x)=x+x^p+x^{p^2}+\cdots+x^{p^{t-1}}.\]
We simply denote $\Tr_{\F_q/\F_p}$  by $\Tr$ if there is no confusion in the context.
For any given $\Ff_p$-basis $\{\eta_1,\eta_2,\cdots,\eta_t\}$ of $\Ff_q$, it is well known that there exists a dual $\Ff_p$-basis $\{\theta_1,\theta_2,\cdots,\theta_t\}$ (see \cite{Nie}), i.e,
$$
\Tr(\eta_i\theta_j)=\left\{
            \begin{array}{ll}
              1,& \hbox{for\ $i=j$,}\\
              0,& \hbox{for\ $i\neq j$.}
            \end{array}
          \right.
$$

Thus, one can express any element of $\Ff_q$ by a linear combination of trace functions. Precisely speaking, for any element $\alpha\in\Ff_q$, let $\alpha=\sum_{i=1}^ta_i\theta_i$ with $a_i\in\Ff_p$. Then we have
$$
\Tr(\alpha\eta_j)=\sum_{i=1}^ta_i\Tr(\theta_i\eta_j)=a_j,
$$
for any $j\in[t]$, i.e,
\begin{equation}\label{S21}
\alpha=\sum_{i=1}^t\Tr(\alpha\eta_i)\theta_i.
\end{equation}
The above trace representation plays a key role in our framework of repairing the failed node.

For a subset $S=\{\gamma_1,\gamma_2,\cdots,\gamma_n\}$ of $\Ff_q$, denote by ${\rm Span}_{\Ff_p}\{\gamma_1,\gamma_2,\cdots,\gamma_n\}$ the linear span over $\Ff_p$ generated by $S$, i.e.,
$$
{\rm Span}_{\Ff_p}\{\gamma_1,\gamma_2,\cdots,\gamma_n\}=\left\{\sum_{z\in S}a_zz:a_z\in \F_p\right\}.
$$
It is easy to see that ${\rm Span}_{\Ff_p}\{\gamma_1,\gamma_2,\cdots,\gamma_n\}$ and ${\rm Span}_{\Ff_p}\{\alpha\gamma_1,\alpha\gamma_2,\cdots,\alpha\gamma_n\}$ have the same dimension over $\Ff_p$ for any nonzero element $\alpha\in\Ff_q$.

Suppose that $V$ is an $\Ff_p$-subspace of $\Ff_q$. Let $L_V(x)$ be the $p$-linearized polynomial defined by
$$
L_V(x)=\prod_{\beta\in V}(x-\beta).
$$
Then $L_V(x)$ is an $\Ff_p$-linear map from $\Ff_q$ to $\Ff_q$ given by $\alpha\mapsto L_V(\alpha)$. Clearly, the kernel of $L_V(x)$ is $V$ and the image ${\rm Im}(L_V)$ forms an $\Ff_p$-subspace of $\Ff_q$ with dimension equal to $t-{\rm dim}_{\Ff_p}(V)$.

\subsection{Background on Reed-Solomon codes}

Let $\cc$ be an $[n,k,d]$ linear code over $\Ff_q$ with dimension $k$ and minimum Hamming distance $d$. The linear code $\cc$ is called an MDS code if $n=k+d-1$. For a $q$-ary $[n,k,d]$ linear code $\cc$ , the dual of $\cc$ is defined by the set
$$
\cc^{\bot}=\left\{(x_1,x_2,\cdots,x_n)\in\Ff_q^n:\sum_{i=1}^nc_ix_i=0,\ {\rm for\ all }\ (c_1,c_2,\cdots,c_n)\in\cc\right\}.
$$

Let $\alpha_1,\alpha_2,\cdots,\alpha_n$ be $n$ distinct elements of $\Ff_q$, where $1<n\leq q$. For $n$ nonzero fixed elements $v_1,v_2,\cdots,v_n$ of $\Ff_q$ ($v_i$ may not be distinct), the generalized Reed-Solomon code \cite{Mac} associated with $\mathbf{a}=(\alpha_1,\alpha_2,\cdots,\alpha_n)$ and $\mathbf{v}=(v_1,v_2,\cdots,v_n)$ is defined by
\begin{equation}\label{eq1}
\GRS_k(\mathbf{a},\mathbf{v})=\left\{(v_1f(\alpha_1),v_2f(\alpha_2),\cdots,v_nf(\alpha_n)):f(x)\in\Ff_q[x], \ {\rm deg}(f(x))\leq k-1\right\}.
\end{equation}
The code $\GRS_k(\mathbf{a},\mathbf{v})$ is a $q$-ary $[n,k,n-k+1]$-MDS code \cite[Th. 9.1.4]{Mac}. If $\mathbf{v}=(1,1,\cdots,1)$, then the generalized Reed-Solomon code is termed as a  Reed-Solomon code. The elements $\alpha_1,\alpha_2,\cdots,\alpha_n$ are called the evaluation points of $\GRS_k(\mathbf{a},\mathbf{v})$. It is well known that the dual of $\GRS_k(\mathbf{a},\mathbf{1})$ is $\GRS_{n-k}(\mathbf{a},\mathbf{u})$ \cite{Jin}, where $\mathbf{1}$ stands for the all-one row vector of length $n$ and $\mathbf{u}=\{u_1,u_2,\cdots,u_n\}$ with $u_i=\prod_{1\leq j\leq n,j\neq i}(\alpha_i-\alpha_j)^{-1}$ for $1\leq i\leq n$.

\subsection{Problem formulation}

Assume that a file of size $M=tk\log p$ is divided into $k$ blocks  and each block stores an element of $\F_q$. Then these $k$ blocks are  encoded using an MDS code  $\mC$  of length $n$  with each node  storing an element of $\F_q$.   The information in each node can be considered as an element of $\F_q$ that is a $t$-dimensional vector over $\F_p$. We suppose that the encoded erasure code is an MDS code in this paper, i.e., it  achieves the minimum storage property.

In the rack-aware model, we further assume that the nodes are organized into groups of size $u$. Thus there are $r=n/u$ racks in total ($n$ is a multiple of $u$). Assume that the transmission cost among the nodes in the same rack is negligible, i.e., cross-rack bandwidth is the major concern. Data within a rack can be computed before being sent to the failure node. To repair a failed node in a rack, the system downloads data from the other nodes in the host rack as well as that in other arbitrary $d$ ($d\le r-1$) helper racks. The information within each rack can be processed before sending for repairing the failed node. Then the failed node can be regenerated by downloading $s$ symbols of $\F_p$ in each of the $d$ help racks as well as $(u-1)t$ symbols in the host rack. Thus, the cross-rack repair bandwidth is $ds\log p$. An encoding scheme that satisfies all the above parameters $n,k,d,t,s$ is called a {\it rack-based storage system} $RSS(n,k,r,d,t,s)$. If such a scheme satisfies the equality in \cite[Theorem 1]{HL18}, then it is called a {\it rack-aware regenerating code  $RRC(n,k,r,d,a,b)$}, where $a=\log q$ is data size stored in each node and $b$ is the bandwidth required to repair a failure node.

In this paper, we only consider the exact repair of rack-aware regenerating codes for one failed node. Furthermore,  the encoded erasure code is taken to be an MDS code in this paper, i.e., it achieves the minimum storage property.  An MSRR code  $RRC(n,k,r,d,a,b)$ is called {\it homogeneous}  if $kr/n$ is an integer, otherwise it is called  {\it hybrid}. As a homogeneous MSRR code can be obtained via a MSR code, we focus on the construction of hybrid MSRR codes in the following context.

  The rack-aware version of the cut-set bound for  rack-aware regenerating codes   is given in \cite{Hu16} through the information flow graph.
  It says that, to repair a failure node for a rack-aware regenerating code  $RRC(n,k,r,d,a,b)$, one has to download  $b$ bits from helper racks with
\[b\geq\frac{d\log q}{d-\left\lfloor kr/n\right\rfloor+1}.\]
In the sequel, we call this bound the rack-aware cut-set bound. If $u=n/r=1$, this bound is degenerated to the cut-set bound for MSR codes.

\section{A general repair framework for Reed-Solomon codes}
In this section, we present a general framework to repair Reed-Solomon codes in a rack-based storage system.

For any two positive integers $i\leq j$, we denote $[i,j]=\{i,i+1,\cdots,j\}$. Assume that $p$ is a prime power and $q=p^t$.

Let $\mathbf{a}=(\alpha_{1,1}, \alpha_{1,2},\dots, \alpha_{1,u},\dots, \alpha_{r,1}, \alpha_{r,2},\dots, \alpha_{r,u})\in\F_q^n$ with $n=ur$ and $\alpha_{i,j}$ being distinct elements of $\F_q$.
Let $\GRS_k(\mathbf{a},\mathbf{1})$ be an $[n=ur,k]$ Reed-Solomon code over $\Ff_q$. Recall that in a rack-based system model, $n$ nodes are divided into $r$ racks. Each codeword $(f(\alpha_{1,1}), f(\alpha_{1,2}),\dots, f(\alpha_{1,u}),\dots, f(\alpha_{r,1}), f(\alpha_{r,2}),\dots, f(\alpha_{r,u}))$ of $\GRS_k(\mathbf{a},\mathbf{1})$ with $f\in\F_q[x]_{<k}$  can be equally divided into $r$ racks, with $u$ nodes in each rack as follows:
\begin{equation}\label{eq:I}
\begin{split}
{\rm Rack}\ 1:& \; f(\alpha_{1,1}),f(\alpha_{1,2}),\cdots,f(\alpha_{1,u}),\\
{\rm Rack}\ 2:&\;  f(\alpha_{2,1}),f(\alpha_{2,2}),\cdots,f(\alpha_{2,u}),\\
&\; \ \ \ \ \ \ \ \ \ \ \ \ \vdots\\
{\rm Rack}\ r:&\;  f(\alpha_{r,1}),f(\alpha_{r,2}),\cdots,f(\alpha_{r,u}).
\end{split}
\end{equation}
This means the data at position $(j,i)$ is the evaluation of the polynomial $f(x)$ on the evaluation point $\alpha_{j,i}$.  Now we give a linear repairing scheme for Reed-Solomon codes in the rack-based model. 

\begin{thm}\label{thm1}Assume that the Reed-Solomon code $\GRS_k(\mathbf{a},\mathbf{1})$ is encoded for a rack-based system model with $r$ racks.
Let $\{\eta_1,\eta_2,\cdots,\eta_t\}$ be a basis of $\Ff_q$ over $\Ff_p$. Let $1\le d\le r-1$.
 For a given pair $(s,j)\in [r]\times[u]$, if there exist $t$ polynomials $h_a(x)\in\Ff_q[x]$ of degree at most $u(d+1)-k-1$ such that $h_a(\alpha_{s,j})=\eta_a$ for any $a\in[t]$, then we can repair the $j$-th node of the $s$-th rack with the help of any other $d$ racks and the cross-rack repair bandwidth is given by
$
b={\rm max}_{S\subseteq[r]\setminus \{s\},|S|=d}\sum_{i\in S}b_i\log p,
$
where
$$
b_i={\rm dim}_{\Ff_p}\left({\rm Span}_{\Ff_p}\left\{\left(h_a(\alpha_{i,1}),h_a(\alpha_{i,2})\dots, h_a(\alpha_{i,u})\right):a=1,2\cdots,t\right\}\right).
$$
\end{thm}
\begin{proof}
Assume that the $s$-th rack  is the host rack and the data at node $(s,j)$  is erased. We aim to repair $f(\alpha_{s,j})$.

Let $S$ be a subset of $[r]\setminus \{s\}$ of size $d$. Consider the following punctured code of $\GRS_k(\mathbf{a},\mathbf{1})$
$$
\cc_d=\left\{(f(\alpha_{i,j}))_{i\in \{s\}\cup S,\ 1\leq j\leq u}:f(x)\in \Ff_q[x],\ {\rm deg}(f(x))\leq k-1\right\}.
$$
The dual code of $\cc_d$ is
$$
\cc_d^{\perp}=\left\{(v_{i,j}g(\alpha_{i,j}))_{i\in \{s\}\cup S,\ 1\leq j\leq u}:g(x)\in \Ff_q[x],\ {\rm deg}(g(x))\leq u(d+1)-k-1\right\}$$
for some $v_{i,j}\in \Ff_q^*$.  Since ${\rm deg}(h_a(x))\leq u(d+1)-k-1$, we have
$$
0=\sum_{i\in \{s\}\cup S}\sum_{j=1}^uv_{i,j}h_a(\alpha_{i,j})f(\alpha_{i,j}),
$$
for any $a\in[t]$. This gives
\begin{equation}\label{peq1}
\Tr(\eta_af(\alpha_{s,j}))=-\sum_{1\le \ell\le u,\ell\neq j}\Tr\left(\frac{v_{s,\ell}}{v_{s,j}}h_a(\alpha_{s,\ell})f(\alpha_{s,\ell})\right)-\sum_{i\in S}\Tr\left(\sum_{\ell=1}^u\frac{v_{i,\ell}}{v_{s,j}}h_a(\alpha_{i,\ell})f(\alpha_{i,\ell})\right).
\end{equation}
Suppose that $\{\bc_1,\bc_2,\cdots,\bc_{b_i}\}$ is an $\Ff_p$-basis of ${\rm Span}_{\Ff_p}\{\left(h_a(\alpha_{i,j})\right)_{1\leq j\leq u}:a=1,2\cdots,t\}$, where $i\in S$. Then, for any $(i,a)\in S\times [t]$, there exist some $e_{\ell,i,a}\in \Ff_p$ with $1\le \ell\le b_i$ such that
$$
(h_a(\alpha_{i,1}),h_a(\alpha_{i,2}),\cdots,h_a(\alpha_{i,u}))=\sum_{\ell=1}^{b_i}e_{\ell,i,a}\bc_\ell.
$$
It follows from (\ref{peq1}) that
$$
\Tr(\eta_af(\alpha_{s,j}))=\sum_{1\le \ell\le u,\ell\neq j}\Tr\left(\frac{v_{s,\ell}}{v_{s,j}}h_a(\alpha_{s,\ell})f(\alpha_{s,\ell})\right)-\sum_{i\in S}\sum_{\ell=1}^{b_i}e_{\ell,i,a}\Tr(\mathbf{f}_i\cdot\bc_\ell),
$$
where $\mathbf{f}_i=\left(\frac{v_{i,1}}{v_{s,j}}f(\alpha_{i,1}),\frac{v_{i,2}}{v_{s,j}}f(\alpha_{i,2}),\cdots,\frac{v_{i,u}}{v_{s,j}}f(\alpha_{i,u})\right)$ and $\mathbf{f}_i\cdot\bc_\ell$ denotes the inner product of $\mathbf{f}_i$ with $\bc_\ell$. Thus, by downloading
$$
\left\{\Tr\left(\frac{v_{s,\ell}}{v_{s,j}}h_a(\alpha_{s,\ell})f(\alpha_{s,\ell})\right):\; \ell\in[u]\setminus\{j\}, \quad a\in[t]\right\}
$$
from the host rack and $\{\Tr(\mathbf{f}_i\cdot\bc_\ell):\; i\in S,\ell\in[b_i]\}$ from $d$ helper racks, we can regenerate
$\Tr(\eta_af(\alpha_{s,j}))$ for all $a\in[t]$. Hence, $f(\alpha_{s,j})$ is repaired.
It is easy to see that one needs to download at most $\sum_{i\in S}b_i\log p$ bits data in total from $d$ helper racks.
\end{proof}

Note that rack-aware regenerating codes can be reduced to regenerating codes if each rack contains only one node, i.e., $r=n$. In this case, we get a linear repairing scheme for Reed-Solomon codes from Theorem \ref{thm1}.  The ideal of downloading a few elements from a small subfield instead of a single element from $\F_q$ was initially proposed by Guruswami et al. \cite{Guru}. This technique was further extended in \cite{Tamo,JLX}.

\begin{cor}\label{cor11}
Assume that the Reed-Solomon code $\GRS_k(\mathbf{a},\mathbf{1})$ with $\mathbf{a}=(\alpha_1,\alpha_2,\dots,\alpha_n)$ is encoded for a standard model, i.e., a degenerated rack-based system with $n=r$ and $u=1$.
Let $\{\eta_1,\eta_2,\cdots,\eta_t\}$ be a basis of $\Ff_q$ over $\Ff_p$. Let $1\le d\le n-1$. For a given $j\in [n]$, if there exist $t$ polynomials $h_a(x)\in\Ff_q[x]$ of degree at most $d-k$ such that $h_a(\alpha_{j})=\eta_a$ for any $a\in[t]$, then we can repair the $j$-th node with the help of any other $d$ nodes and the bandwidth is
$
b={\rm max}_{S\subseteq[n]\setminus \{j\},|S|=d}\sum_{i\in S}b_i\log p,
$
where
$$
b_i={\rm dim}_{\Ff_p}\left({\rm Span}_{\Ff_p}\left\{h_a(\alpha_i):\; a=1,2\cdots,t\right\}\right).
$$
\end{cor}

\section{Interpreting existing Examples under our general framework}

In this section, we illustrate that the repairing schemes given in \cite{Guru,Tamo,CB19} can be realized under our general framework introduced in Section III.

Let us first look at the two examples given by Guruswami et al. \cite{Guru}. They presented two constructions of RS codes and their corresponding repairing schemes. In the following, we show that the repairing schemes in \cite{Guru} can be realized in the framework of Corollary \ref{cor11}.

\begin{prop}{\rm  \cite[Theorem 1]{Guru}}
Let $p$ be a prime power and $\F_q/\F_p$ be a field extension with $[\F_q:\F_p]=t$. Label $\Ff_q=\{\alpha_1,\alpha_2,\dots,\alpha_q\}$. Suppose that the Reed-Solomon code $\GRS_k(\mathbf{a},\mathbf{1})$ over $\Ff_q$ with $\mathbf{a}=(\alpha_1,\alpha_2,\dots,\alpha_q)$ is encoded for a standard model. If $k\le q(1-1/p)$, then one can repair the $j$-th node with the help of any other $q-1$ nodes and the bandwidth is $(q-1)\log p$.
\end{prop}
\begin{proof}
We assume that data at the $j$-th position is erased. Let $\{\eta_1,\eta_2,\cdots,\eta_t\}$ be a basis of $\Ff_q$ over $\Ff_p$. For any $a\in[t]$, one can take $h_a(x)$ to be the polynomial
\[h_a(x)=\frac{\Tr(\eta_a(x-\alpha_j))}{x-\alpha_j}.\]
It is easy to verify that ${\rm deg}(h_a(x))\leq q-1-k$ and $h_a(\alpha_{j})=\eta_a$ for any $a\in[t]$. Note that
$$
{\rm Span}_{\Ff_p}\left\{h_a(\alpha_i):\; a=1,2\cdots,t\right\}\subseteq \frac{1}{\alpha_i-\alpha_j}\Ff_p
$$
for any $i\neq j$.
By Corollary \ref{cor11}, we can repair the $j$-th node with the help of any other $q-1$ nodes and the total bandwidth is
$
b=(q-1)\log p,
$
\end{proof}

Now we consider the second example given in \cite{Guru}.
\begin{prop}{\rm  \cite[Theorem 10]{Guru}} Let  $n\le 2(2^s-1)$ be an even number.
Assume that $q=2^{2s}$ and the evaluation points set $\{\alpha_1,\cdots,\alpha_n\}$  is consisted of $n/2$ points from $\Ff_{2^s}^*$ and $n/2$ points from $\beta\Ff_{2^s}^*$, where $\beta$ is a primitive element of $\Ff_q$. If  $k\leq n-2$, then the Reed-Solomon code $\GRS_k(\mathbf{a},\mathbf{1})$ over $\Ff_q$ admits a linear exact repair scheme for a standard model with bandwidth $s(\frac{3n}{2}-2)$.
\end{prop}
\begin{proof}
Define $h_1(x)=1$ and
$$
h_2(x)=\left\{
            \begin{array}{ll}
              x,& \hbox{if\ $\alpha_j\in\beta\Ff_{2^s}^*$,}\\
             \beta^{-1}x,& \hbox{if $\alpha_j\in\Ff_{2^s}^*$,}
            \end{array}
          \right.
$$
where the $j$-th position is erased. It is easy to see that ${\rm deg}(h_a(x))\leq n-k-1$. If $\alpha_j\in\Ff_{2^s}^*$, then $h_1(\alpha_j)=1$ and $h_2(\alpha_j)=\beta^{-1}\alpha_j$. Note that $\{1,\beta^{-1}\alpha_j\}$ is a basis of $\Ff_{2^{2s}}$ over $\Ff_{2^s}$. It follows from Corollary \ref{cor11} that one can repair the $j$-th node with the help of the other $n-1$ nodes and the bandwidth is $(\frac{3n}{2}-2)\log 2^s$. The case where $\alpha_j\in\beta\Ff_{2^s}^*$ can be proved similarly.
\end{proof}

Due to some limitations of RS codes, it seems very difficult to propose a linear repair scheme for RS codes under a standard model that meets the cut-set bound with equality. To our best of knowledge, there is only one construction of RS codes achieving the cut-set bound \cite{Tamo}. Again we can show that  RS codes in \cite{Tamo} can be repaired under our general framework of Corollary \ref{cor11}.

\begin{prop}{\rm  \cite[Theorem 7]{Tamo}}
Let $n>d>k$ be positive integers. Then there exists a RS code meeting the cut-set bound for the repair of any single node from any $d$ helper nodes. Furthermore, the alphabet size of the RS code is $\Omega(\exp(n^n))$.
\end{prop}
\begin{proof}
Suppose that $p$ is a prime power and $s=d-k+1$. By Dirichlet's Theorem, there are infinitely many primes $\ell$ with $\ell\equiv 1 \pmod {s}$. Choose distinct primes $\ell_1,\ell_2,\cdots,\ell_n$ such that $\ell_i\equiv1\pmod {s}$ and $\alpha_i\in\overline{\Ff}_p$ such that $[\Ff_p(\alpha_i):\Ff_p]=\ell_i$. Define
$$
F=\Ff_p(\alpha_1,\cdots,\alpha_n),\ \ F_i=\Ff_p(\alpha_j:j\neq i).
$$
Put $p_i=|F_i|=p^{\prod_{j\neq i}l_j}$. Assume that $\Ff_q$ is an extension of $F$ with degree $s$. Then $q=n^{\Omega(n^2)}$. Take the evaluation set of the RS code over $\Ff_q$ to be $\{\alpha_1,\cdots,\alpha_n\}$. Assume that data at the $j$-th position is erased. Due to the existence of an $F_j$-subspace of $\Ff_q$ such that
$$
{\rm dim}_{F_j}S_j=\ell_j,\ \ S_j+S_j\alpha_j+\cdots+S_j\alpha_j^{s-1}=\Ff_q,
$$
one can define the polynomial
$$
h_a(x)=\beta_vx^c,
$$
for any $a=c\ell_j+v$ with $1\leq v\leq \ell_j$ and $0\leq c\leq s-1$ ,where $\beta_1,\cdots,\beta_{\ell_j}$ is an $F_j$-basis of $S_j$. Since $\{\beta_u\alpha_j^v:u\in[\ell_j],v=0,1,\cdots,s-1\}$ is an $F_j$-basis of $\Ff_q$ and ${\rm deg}(h_a(x))\leq s-1=d-k$, by Corollary \ref{cor11}, the $j$-th position can be repaired with the help of any other $d$ nodes by downloading $b={\rm max}_{S\in[n]\setminus \{j\}, |S|=d}\sum_{i\in S}b_i\log p_j$ bits of data, where
$$
b_i={\rm dim}_{F_j}{\rm Span}_{\Ff_j}\left\{h_a(\alpha_i):\; a=1,2\cdots,s\ell_j\right\}={\rm dim}_{F_j}S_j=\ell_j
$$
since $h_a(\alpha_i)=\beta_vx^c\in S_j$. Hence the bandwidth is at most $d\ell_j\log p_j$. The cut-set bound shows that
$$
b\geq\frac{d\log q}{d-k+1}=\frac{d\log p_j^{s\ell_j}}{s}=d\ell_j\log p_j.
$$
Thus, the cut-set bound is achieved.
\end{proof}

Very recently, by modifying the construction of RS codes in \cite[Theorem 7]{Tamo}, Chen and Barg {\rm \cite[Proposition V.2 of Section 5]{CB19}} proposed a construction of RS codes that meet the rack-aware cut-set bound with equality. Now, we interpret their repairing scheme under our general framework of Theorem \ref{thm1}.

Let $q$ be a power of a prime. Suppose that $u$ is the size of the rack with $u\mid(q-1)$. Let $k=mu+v$ and $s=d-m+1$, where $0\leq v\leq u-1$. Assume that $\ell_1,\ell_2,\cdots,\ell_r$ are distinct primes such that $\ell_i\equiv 1\pmod {s}$ and $\ell_i>u$ for any $i\in[r]$. Choose $\alpha_i\in\overline{\Ff}_q$ such that $[\Ff_q(\alpha_i):\Ff_q]=\ell_i$. Define
$$
F=\Ff_q(\alpha_1,\cdots,\alpha_r),\ \ F_i=\Ff_q(\alpha_j:j\neq i).
$$
Then $[F:\Ff_q]=\prod_{i=1}^r\ell_i$ and $[F:F_i]=\ell_i$. Put $p_i=|F_i|=q^{\prod_{j\neq i}\ell_j}$. Assume that $K$ is an extension of $F$ with degree $s$. Further, suppose that $\alpha\in\Ff_q$ is an element of multiplicative order $u$. Consider the set of elements
$$
\alpha_{i,j}=\alpha_{i}\alpha^{j-1}, i=1,\cdots,r,\ \ j=1,\cdots,u.
$$

Let $GRS_k(\mathbf{a},\mathbf{1})$ be the generalized RS code with parameters $[n=ur,k]$ over the finite field $K$ defined in the rack-based model \eqref{eq:I}. To state the repair procedure, the following lemma is necessary.

\begin{lem}{\rm \cite{CB19}}\label{lem1}
For any $i\in[r]$, there exists an $F_i$-subspace of $K$ such that
$$
{\rm dim}_{F_i}S_i=\ell_i,\ \ S_i+S_i\alpha_i^u+\cdots+S_i\alpha_i^{u(s-1)}=K,
$$
where the operation $+$ is the Minkowski sum of sets, $T_1+T_2:=\{\beta_1+\beta_2:\beta_1\in T_1,\beta_2\in T_2\}$.
\end{lem}

Assume that the Reed-Solomon code $\GRS_k(\mathbf{a},\mathbf{1})$ is encoded for a rack-based system model with $r$ racks.
\begin{prop}{\rm  \cite[Proposition V.2]{CB19}}\label{prop:4.5}
With the help of any $d$ helper racks, the Reed-Solomon code $\GRS_k(\mathbf{a},\mathbf{1})$ defined above admits optimal repairing scheme with respect to the rack-aware cut-set bound. Furthermore, the alphabet size of the RS code is $\Omega(\exp(n^n))$.
\end{prop}
\begin{proof}
Let the index of the host rack be $j$ and the index of the failed node in this rack be $p$, where $1\leq j\leq r$ and $1\leq p\leq u$. Consider the field extension $K$ of $F_j$. Note that $[K:F_j]=s\ell_j$. Let $\beta_1,\cdots,\beta_{\ell_j}$ be an $F_j$-basis of $S_j$. For $a=c\ell_j+v$ with $1\leq v\leq \ell_j$ and $0\leq c\leq s-1$, put $h_a(x)=\beta_vx^{uc}$. Then
\begin{eqnarray*}
{\rm deg}(h_a(x)) &\leq& u(s-1)
=u(d-m)\\
&\leq&u(d+1)-um-v-1=u(d+1)-k-1.
\end{eqnarray*}
Observe that $h_a(\alpha_{j,p})=\beta_v\alpha_j^{uc}$ for any $1\leq v\leq \ell_j$ and $0\leq c\leq s-1$. By Lemma \ref{lem1}, $\{\beta_v\alpha_j^{uc}:1\leq v\leq \ell_j,\ 0\leq c\leq s-1\}$ forms a $F_j$-basis of $K$. It follows from Theorem \ref{thm1} that we can repair the $p$-th node of the $j$-th rack with the help of any other $d$ racks by downloading $b={\rm max}_{S\subseteq[r]\setminus \{j\},|S|=d}\sum_{w\in S}b_w\log p_j$ bits of data from any other $d$ helper racks, where
$$
b_w={\rm dim}_{F_j}\left({\rm Span}_{F_j}\left\{\left(h_{a}(\alpha_{w,j})\right)_{1\leq j\leq u}:a=1,2\cdots,sl_j\right\}\right).
$$
Note that
$$
{\rm Span}_{F_j}\left\{\left(h_{a}(\alpha_{w,j})\right)_{1\leq j\leq u}:a=1,2\cdots,sl_j\right\}\subseteq S_j\cdot(1,1,\cdots,1).
$$
Hence, the cross-rack repair bandwidth is at most $b=d\ell_j\log p_j$. The corresponding rack-aware cut-set bound is
$$
b\geq \frac{d\log |K|}{d-m+1}=\frac{d\log p_j^{sl_j}}{s}=dl_j\log p_j.
$$
Therefore, the rack-aware repairing scheme for any rack $j$ is optimal.
\end{proof}

\section{Reed-Solomon codes meeting the rack-aware cut-set bound}
The Reed-Solomon codes given in Proposition \ref{prop:4.5} have alphabet size exponential in the code length $n$. In this section, we construct several classes of   Reed-Solomon codes achieving the rack-aware cut-set bound that have alphabet size linear in $n$.

Theorem \ref{thm1} provides a framework for repairing RS codes. More precisely, to find a repairing scheme for RS codes with one node failure, one needs to find polynomials $h_a(x)$ satisfying the conditions in Theorem \ref{thm1}. In this section, we show via the above framework that RS codes can be repaired with the cross-rack repair bandwidth meeting the rack-aware cut-set bound.

Firstly let us define what is a good polynomial. Such good polynomials have been  used for constructions of locally repairable codes, see \cite{LRC1}. Here we make use of good polynomials for the construction of MSRR codes.
\begin{lem}\label{lem:4.1}
Assume that $A_1,A_2,\cdots,A_r$ are $r$ pairwise disjoint subsets of $\F_q$ with $|A_i|=u$ for $i=1,2,\dots,r$. Put $n=ur$ and $A_i=\{\Ga_{i,1},\Ga_{i,2},\cdots,\Ga_{i,u}\}$. A polynomial $g(x)\in\F_q[x]$  of degree $u$ is called a good polynomial if  $g(x)$ is equal to 0 in $A_s$ for some $s\in [r]$
and is constant in each $A_i$ for any $i\in[r]\setminus s$.
\end{lem}

Let $\F_p$ be a subfield of $\F_q$ and let $V$ be an $\F_p$-subspace of $\F_q$. Define  the linearized polynomial $L_V(x)=\sum_{\Ga\in V}(x-\Ga)$. Suppose that $\{\eta_1,\eta_2,\cdots,\eta_t\}$ is a basis of $\Ff_q$ over $\Ff_p$.
Now we use the degree reduction technique to get a general construction.

\begin{thm}\label{thm:4.1} Let $p$ be a prime power and let $q=p^t$ for some integer $t\ge 2$. Let $g(x)\in\F_q[x]$ be a good polynomial defined as lemma \ref{lem:4.1}.
Let $V$ be an $\F_p$-subspace of dimension $\ell$ of $\F_q$ and define the linearized polynomial $L_V(x)=\sum_{\Ga\in V}(x-\Ga)$ and define the polynomial $h_a(x):=\frac{L_V(g(x)\eta_a)}{g(x)} \pmod{x^q-x}.$ If $w\le u(d+1)-k-1$ with $w=\max\{\deg(h_a(x)):\; a\in [t]\}$, then one can repair any node in the $s$-th rack for the generalized RS code $\GRS_k(\mathbf{a},\mathbf{1})$ defined in the rack-based model \eqref{eq:I} with the cross-rack repair bandwidth $d(t-\ell)\log p$ bits data.
\end{thm}
\begin{proof}  Define the polynomial $h_a(x):=\frac{L_V(g(x)\eta_a)}{g(x)}$.  It is easy to see that $h_a(\Ga_{s,j})=\eta_a$ for $\Ga_{s,j}\in \F_q$. Since the degree of $h_a(x)$ is at most $w$, by Theorem \ref{thm1} we can repair any node in the $s$-th rack by downloading at most $
b={\rm max}_{S\subseteq[r]\setminus \{s\},|S|=d}\sum_{i\in S}b_i\log p
$ bits data,
where
$$
b_i={\rm dim}_{\Ff_p}\left({\rm Span}_{\Ff_p}\left\{\left(h_a(\alpha_{i,1}),h_a(\alpha_{i,2})\dots, h_a(\alpha_{i,u})\right):a=1,2\cdots,t\right\}\right).
$$
 For any $i\in S$, we have
\begin{eqnarray*}
&&{\rm Span}_{\Ff_p}\left\{\left(h_a(\alpha_{i,1},h_a(\alpha_{i,2}),\dots, h_a(\alpha_{i,u}\right):\; a=1,2\cdots,t\right\}\\
&=&{\rm Span}_{\Ff_p}\left\{\left(\frac{L_V(g(\alpha_{i,1})\eta_a)}{g(\alpha_{i,1})},\cdots,\frac{L_V(g(\alpha_{i,u})\eta_a)}{g(\alpha_{i,u})}\right):a=1,2\cdots,t\right\}\\
&=&{\rm Span}_{\Ff_p}\left\{\left(\frac{L_V(g(\alpha_{i,1})\eta_a)}{g(\alpha_{i,1})},\cdots,\frac{L_V(g(\alpha_{i,u})\eta_a)}{g(\alpha_{i,1})}\right):a=1,2\cdots,t\right\}\\
&\subseteq&\left(\frac{1}{g(\alpha_{i,1})},\cdots,\frac{1}{g(\alpha_{i,1})}\right)L_V(\Ff_q).
\end{eqnarray*}
As $L_V(x)$ induces an $\F_p$-linear map from $\F_q$ to $\F_q$ with kernel $V$, the space  $L_V(x)$ has $\F_p$-dimension $t-\ell$. This implies that $b_i\le t-\ell$. Hence,
$b={\rm max}_{S\subseteq[r]\setminus \{s\},|S|=d}\sum_{i\in S}b_i\log p\le d(t-\ell)\log p$. The proof is completed.
\end{proof}

By choosing different good polynomials, we are going to present several explicit constructions of RS codes meeting the rack-aware cut-set bound.
Defining a good polynomial $g(x)$ from an additive subgroup of $\Ff_q$, we derive the following result.

\begin{thm}\label{thm:4.2}
Let $p$ be a prime power. Assume that $t>2$ is an even positive integer. Put $n=q=p^t$ and we equally divide $n$ nodes into $r=p^{2}$ racks such that each rack contains $u=p^{t-2}$ nodes. Let $m=\lfloor k/p^{t-2}\rfloor$ and $\ell<t$ be a positive integer. If $k\leq p^t-p^{t-1}+p^{t-2}-1$ and $p^{2}-m=\frac{t}{t-\ell}$, then there exist an optimal rack-aware repairing scheme for $GRS_k(\mathbf{a},\mathbf{1})$ with respect to the rack-aware cut-set bound.
\end{thm}

\begin{proof}
Let $W=\{x\in\Ff_{p^t}:\Tr_{\Ff_{p^t}/\Ff_{p^2}}(x)=0\}$. Put $r=p^{2}$. Let $A_1,A_2,\dots,A_{r}$ be $r=p^{2}$ pairwise distinct cosets of $W$ in $\F_q$. Define the polynomial $g(x)=\Tr_{\Ff_{p^t}/\Ff_{p^2}}(x)-\Tr_{\Ff_{p^t}/\Ff_{p^2}}(\beta)$ for some $\beta\in A_s$ and $h_a(x):=\frac{L_V(g(x)\eta_a)}{g(x)}$. Then $g(x)$ is constant in each $A_i$ and $g(x)$ is equal to $0$ in $A_s$. It is easy to check that
 ${\rm deg}(h_a(x))=p^{t-2}(p-1)$. Observe that $p^{t-2}(p-1)\leq p^{t-2}(p^{2}-1+1)-k-1$.

Hence, it follows from Theorem \ref{thm:4.1} that the cross-rack repair bandwidth is at most $(p^{2}-1)(t-\ell)\log p$. The rack-aware cut-set bound shows that
$$
b\geq \frac{(p^{2}-1)\log q}{p^{2}-m}=\frac{(p^{2}-1)t\log p}{t/(t-\ell)}=(p^{2}-1)(t-\ell)\log p.
$$
Therefore, the repairing scheme is optimal.
\end{proof}

In the following, we define the polynomial $g(x)$ from a multiplicative subgroup of $\Ff_q^*$. Now we obtain an optimal repairing scheme of  generalized RS codes with respect to the rack-aware cut-set bound.

\begin{thm}\label{thm:4.3}
Let $p$ be a prime power and $q=p^t$ for some integer $t\ge 2$. Suppose that $1\le a<t$ is an integer with $a\mid t$. We set $n=q-1$ and equally divide $n$ nodes into $r=p^{a}-1$ racks such that each rack contains $u=\frac{p^t-1}{p^a-1}$ nodes. Assume that $m=\left\lfloor \frac{k(p^a-1)}{p^t-1}\right\rfloor$ and $\ell<t$ is a positive integer such that $a+\ell>t$ and $p^{a}-1-m=\frac{t}{t-\ell}$. If $k\leq (p^t-p^l)\frac{p^t-1}{p^t-p^{t-a}}-1$, then there exists a repairing scheme for $GRS_k(\mathbf{a},\mathbf{1})$ meeting the rack-aware cut-set bound.
\end{thm}
\begin{proof}
Let $H$ be a multiplicative subgroup of $\Ff_q^*$ of order $u$. Let $A_1,A_2,\dots,A_{r}$ be $r=p^{a}-1$ pairwise distinct cosets of $H$ in $\F_q^*$. Define the polynomial $g(x)=x^u-\beta^u$ for some $\beta\in A_s$ and $h_a(x):=\frac{L_V(g(x)\eta_a)}{g(x)}$. Then $g(x)$ is constant in each $A_i$ and $g(x)$ is equal to $0$ in $A_s$.  It can be easily seen that ${\rm deg}(h_a(x))=p^t-1+(p^l-p^t)\frac{p^t-1}{p^t-p^{t-a}}$. Note that $p^t-1+(p^l-p^t)\frac{p^t-1}{p^t-p^{t-a}}\leq \frac{p^t-1}{p^a-1}(p^{a}-2+1)-k-1$.

Therefore, by Theorem \ref{thm:4.1}, the cross-rack repair bandwidth is at most $(p^{a}-2)(t-\ell)\log p$. The rack-aware cut-set bound shows that
$$
b\geq \frac{(p^{a}-2)\log q}{p^{a}-1-m}=\frac{(p^{a}-2)t\log p}{t/(t-\ell)}=(p^{a}-2)(t-\ell)\log p,
$$
i.e., the repairing scheme is optimal.
\end{proof}

Combining an additive subgroup and a multiplication subgroup of the field $\Ff_{p^t}$, Tamo et al. \cite{LRC1} provided a construction of good polynomials as follows. Let $a,v$ be positive integers such that $a\mid t$ and $p^a\mod v=1$. Assume that $W$ is an additive subgroup of $\Ff_{p^t}$ that is closed under the multiplication by the field $\Ff_{p^a}$. Let $\beta_1,\cdots,\beta_v$ be the $v$-th degree roots of unity in $\Ff_{p^t}$. Then the polynomial
\begin{equation}\label{p1}
G(x)=\prod_{i=1}^v\prod_{w\in W}(x+w+\beta_i)
\end{equation}
is constant on $W$ and the cosets of the union $\cup_{1\leq i\leq v}W+\alpha\beta_i$. That is to say, the field $\Ff_{p^t}$ is partitioned into $(p^t-|W|)/v|W|$ sets of $v|W|$ and one set of size $|W|$ by the polynomial $g(x)$. By using this polynomial, we propose a repairing scheme for the generalized RS codes achieving the rack-aware cut-set bound.

\begin{thm}\label{thm:4.4}
Let $p$ be a prime power and $q=p^t$ for some integer $t\ge 2$. Assume that $a,v,s$ are positive integers such that $a\mid t$, $p\mid\frac{t}{a}$, $v<p$ and $p^a\bmod v=1$. Put $n=q-p^{t-a}$ and equally divide $n$ nodes into $r=\frac{q-p^{t-a}}{vp^{t-a}}$ racks such that each rack contains $u=vp^{t-a}$ nodes. Let $m=\left\lfloor \frac{k}{vp^{t-a}}\right\rfloor$ and $\ell<t$ be a positive integer. If $k\leq p^t-vp^{t-1}-1$ and $\frac{q-p^{t-a}}{vp^{t-a}}-m=\frac{t}{t-\ell}$, then there exists an optimal repairing scheme for $GRS_k(\mathbf{a},\mathbf{1})$ with respect to the rack-aware cut-set bound.
\end{thm}
\begin{proof}
Define the set $W=\{x\in\Ff_{p^t}:\Tr_{\Ff_{p^t}/\Ff_{p^a}}(x)=0\}$. Thanks to $p\mid\frac{t}{a}$, it is easy to check that $W$ is an additive subgroup of $\Ff_{p^t}$ that is closed under the multiplication by the field $\Ff_{p^a}$. Let $G(x)=\prod_{i=1}^v\prod_{w\in W}(x+w+\beta_i)=\left(\Tr_{\Ff_{p^t}/\Ff_{p^a}}(x)\right)^v$ be the polynomial defined by (\ref{p1}) and $A_1,A_2,\dots,A_{r}$ be $r=\frac{q-p^{t-a}}{vp^{t-a}}$ pairwise distinct cosets of $\cup_{1\leq i\leq v}W+\alpha\beta_i$ in $\Ff_q$.

Define the polynomial $g(x)=G(x)-G(\beta)$ for some $\beta\in A_s$ and $h_a(x):=\frac{L_V(g(x)\eta_a)}{g(x)}$. Then $g(x)$ is a good polynomial defined as in lemma \ref{lem:4.1}. And ${\rm deg}(h_a(x))\leq p^{t-a}(vp^{a-1}-1)$. Note that $p^{t-a}(vp^{a-1}-1)\leq  vp^{t-a}(\frac{q-p^{t-a}}{vp^{t-a}}-1+1)-k-1$. It follows from Theorem \ref{thm:4.1} that the cross-rack repair bandwidth is at most $(\frac{q-p^{t-a}}{vp^{t-a}}-1)(t-\ell)\log p$. This achieves the rack-aware cut-set bound with equality.
\end{proof}

\end{document}